\documentclass{IEEEtran}
\onecolumn 
\makeatletter
\usepackage[usenames,dvipsnames]{pstricks}
 \usepackage{epsfig}
 \usepackage{pst-grad} % For gradients
 \usepackage{pst-plot} % For axes
 \usepackage[space]{grffile} % For spaces in paths
 \usepackage{etoolbox} % For spaces in paths
 \makeatletter % For spaces in paths
 \patchcmd\Gread@eps{\@inputcheck#1 }{\@inputcheck"#1"\relax}{}{}

\usepackage[cmex10]{amsmath}
\usepackage{xcolor}
\usepackage[multiple]{footmisc}

\usepackage{mdwtab}
\usepackage{eqparbox}
\usepackage{tabularx}
\usepackage{graphicx,amssymb,rangecite,upgreek,graphicx,dsfont,mathrsfs}
\usepackage{algorithm}
\usepackage{algorithmic} 
\usepackage[usenames,dvipsnames]{pstricks}
 \usepackage{epsfig}
 \usepackage{subcaption}

 \usepackage{pst-grad} % For gradients
 \usepackage{pst-plot} % For axes%\numberwithin{equation}{chapter}

\usepackage{psfrag}

\newcommand {\bE} {\mathbb{E}}

\newcommand {\bX} {\mbox{\boldmath $X$}}

\newcommand{\calL}{{\cal L}}

\newcommand{\calP}{{\cal P}}
\newcommand{\calQ}{{\cal Q}}

\newcommand{\calX}{{\cal X}}
\newcommand{\calY}{{\cal Y}}

\newcommand{\be}{\begin{equation}}
\newcommand{\ee}{\end{equation}}
\newcommand{\beqna}{\begin{eqnarray}}
\newcommand{\eeqna}{\end{eqnarray}}

\usepackage{theorem}
\DeclareFontFamily{U}{mathx}{\hyphenchar\font45}
\DeclareFontShape{U}{mathx}{m}{n}{
      <5> <6> <7> <8> <9> <10>
      <10.95> <12> <14.4> <17.28> <20.74> <24.88>
      mathx10
      }{}
\DeclareSymbolFont{mathx}{U}{mathx}{m}{n}

\DeclareMathSymbol{\bigtimes}{1}{mathx}{"91}
\theorembodyfont{\rmfamily}

\newcommand{\abs}[1]{\left|#1\right|}

\theoremheaderfont{\itshape}

\newtheorem{theorem}{Theorem}
\newtheorem{proof}{Proof}

\newtheorem{lemma}{Lemma} 
\newtheorem{corollary}{Corollary}

\newtheorem{remark}{Remark}
\newcommand{\p}[1]{\left(#1\right)}
\newcommand{\pp}[1]{\left[#1\right]}
\newcommand{\ppp}[1]{\left\{#1\right\}}

\allowdisplaybreaks
\usepackage{comment}
\newcommand{\salman}[1]{\textcolor{black}{#1}}

\newcommand{\changed}[1]{\textcolor{black}{#1}}

\makeatletter
\newcommand{\subalign}[1]{%
	\vcenter{%
		\Let@ \restore@math@cr \default@tag
		\baselineskip\fontdimen10 \scriptfont\tw@
		\advance\baselineskip\fontdimen12 \scriptfont\tw@
		\lineskip\thr@@\fontdimen8 \scriptfont\thr@@
		\lineskiplimit\lineskip
		\ialign{\hfil$\m@th\scriptstyle##$&$\m@th\scriptstyle{}##$\crcr
			#1\crcr
		}%
	}
}
\makeatother

\begin{document}

\title{Why Botnets Work: Distributed Brute-Force Attacks Need No Synchronization}

\author{\IEEEauthorblockN{Salman Salamatian,}
\and
\IEEEauthorblockN{Wasim Huleihel,} 
\and
\IEEEauthorblockN{Ahmad Beirami,} 
\and
\IEEEauthorblockN{Asaf Cohen,}
\and
\IEEEauthorblockN{Muriel M\'edard}
\IEEEoverridecommandlockouts
\IEEEcompsocitemizethanks{
\IEEEcompsocthanksitem
This paper was presented in part at the 2017 IEEE International Symposium on Information Theory \cite{huleihel2017guessing}. Salman Salamatian, Wasim Huleihel, and Muriel M\'edard are with the department of Electrical Engineering and Computer Science, MIT, Cambridge MA (salmansa@mit.edu, wasimh@mit.edu, beirami@mit.edu, medard@mit.edu). A. Beirami was with the department of Electrical Engineering and Computer Science, MIT, Cambridge, MA. He is currently with EA Digital Platform -- Data \& AI, Electronic Arts, Redwood City, CA (ahmad.beirami@gmail.com). Asaf Cohen is with the department of Electrical Engineering, Ben-Gurion University of the Negev, Israel (coasaf@bgu.ac.il).
}}

\parskip 3pt

\maketitle

\begin{abstract}
	In September 2017, McAffee Labs quarterly report \cite{quaterlyreport} estimated that brute force attacks represent 20\% of total network attacks, making them the most prevalent type of attack ex-aequo with browser based vulnerabilities. These attacks have sometimes catastrophic consequences, and understanding their fundamental limits may play an important role in the risk assessment of password-secured systems, and in the design of better security protocols. While some solutions exist to prevent online brute-force attacks that arise from one single IP address, attacks performed by botnets are more challenging. In this paper, we analyze these distributed attacks by using a simplified model. Our aim is to understand the impact of distribution and asynchronization on the overall computational effort necessary to breach a system. Our result is based on Guesswork, a measure of the number of queries (guesses) required of an adversary before a correct sequence, such as a password, is found in an optimal attack. Guesswork is a direct surrogate for time and computational effort of guessing a sequence from a set of sequences with associated likelihoods. 
We model the lack of synchronization by a worst-case optimization in which the queries made by multiple adversarial agents are received in the worst possible order for the adversary, resulting in a min-max formulation. 
We show that, even without synchronization, and for sequences of growing length, the asymptotic optimal performance is achievable by using randomized guesses drawn from an appropriate distribution. Therefore, randomization is key for distributed asynchronous attacks. In other words, asynchronous guessers can asymptotically perform brute-force attacks as efficiently as synchronized guessers.
\end{abstract}

\section{Introduction}\label{sec:intro}
 From online banking \cite{hole2006case} and bitcoin wallets \cite{bitcoinFortune}, to secure shell (SSH), file transfer protocol (ftp), and telnet servers \cite{owens2008study}, and passing by governmental institutions \cite{bruteforceParlement}, brute-force attacks have shown to be one of the major threats to network security. Despite the computational burden on the attacker, brute-force attacks are prevalent. This can be explained through multiple points of view. First, passwords are often weaker than what they ought to be, meaning that attackers can hope to find the correct password well before they query a significant portion of the possible password strings. Next, attacks through huge networks of compromised computers (botnets) are now more common, giving access to significant computational resources for the attacker. More critically, these botnets help to disguise the attack by distributing it.  Indeed, a main solution to the threat of online brute-force attacks is to setup a system that detects and prevents too many queries from any one user, as determined by IP addresses. As such, an attacker which uses only a single IP address would be limited to a fixed number of guesses. In recent years, however, this defense was circumvented by using massive botnets, each bot querying potential passwords. In this situation, it is hard to detect legitimate users in the crowd of illegitimate attackers. These attacks come with a cost, namely, the attack is now distributed across thousands, sometimes millions of computers, each with limited computational power and synchronization tools.
 
As a first step to understand the impact of synchronization, we put forth a simplified mathematical model for passwords and brute-force attacks. We believe that the intuition gained from this model is informative and helpful in assessing the security of systems under brute-force attacks. In particular, we study Guesswork, a measure of the number of password queries (guesses) that an adversary would have to perform before finding the correct one.
Guesswork is best explained through the following simple game: Alice selects a secret discrete random variable $X$ taking values in a finite set $\calX$, and distributed according to $P_X$. 
Then, Bob, who does not see the realization of $X$ but does know $P_X$, presents to Alice a successive sequence of guesses $\hat{X}_1,\hat{X}_2$, and so on.
\salman{For each guess $\hat{X}_i$, Alice checks whether it is the correct symbol $X$.} If the answer is affirmative, Alice says ``yes", and the game ends. Otherwise, the game continues, and Alice examines subsequent guesses. 

This game has a simple interpretation in the context of security. 
Consider a setup where a system is protected using a password $X$, that Alice draws at random from a distribution $P_X$ (or is drawn by nature and revealed to Alice, as it happens in several important password-protection tools which generate passwords, e.g. iCloud keychain). 
An adversary, Bob, wishes to breach the system by performing a brute-force attack, or, in other words, by guessing the password $X$. 
The brute-force attack on the system would consist of, first, producing a list of all possible password strings $\mathcal{X}$ ordered from most, to least likely with respect to $P_X$, and then exhausting the list of passwords one by one until successfully guessing the correct password.
In order to understand the security of such a system under these attacks, it is necessary to evaluate the computational effort required by Bob to breach the system. 
To achieve this, it is reasonable to quantify the number of queries before the correct password is found, which we shall denote by $G^*(X)$, and in particular, its $\rho$-th moment, i.e. $\mathbb{E}[G^*(X)^\rho]$. The number of queries is a direct surrogate for the computational effort that Bob must accomplish, and the lower this quantity, the more vulnerable the system is to brute-force attacks. 

\begin{figure}
    \centering
    \includegraphics[scale=.4]{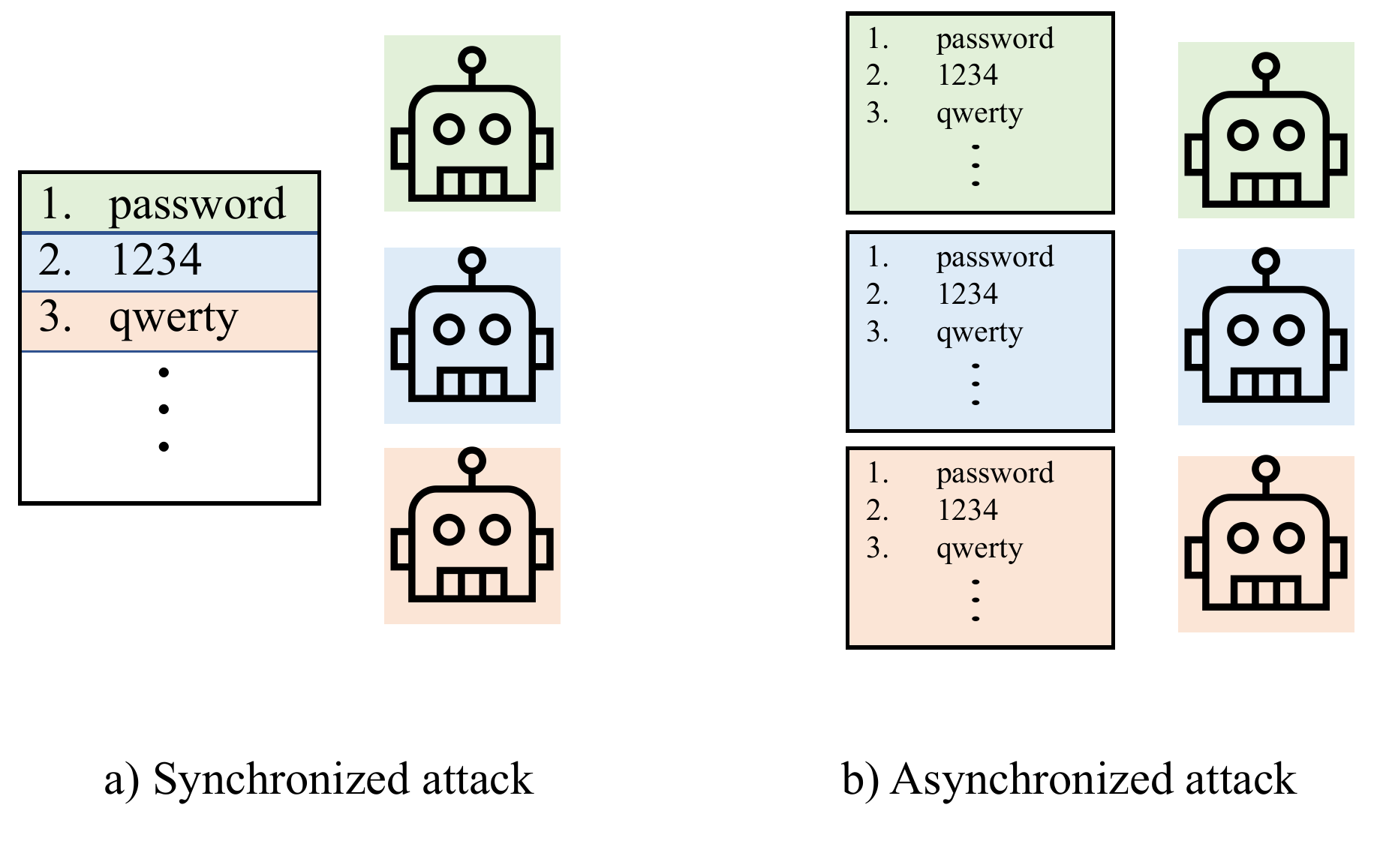}
    \caption{In a synchronized attack, the bots query from the password-list in a specified order. In the asynchronous attack, they do not know the order in which the queries will be sent. Our solution will consist at drawing guesses according to some distribution, instead of querying passwords one-by-one.}
    \label{fig:intro}
\end{figure}

If multiple adversarial agents (we shall use adversary and agents interchangeably) coordinate their attack, the system will be compromised as soon as any of them succeeds.
Moreover, the individual computational effort of each adversary is reduced, while the total number of queries remains the same.
Indeed, an optimal strategy here would consist of having each agent query the most-likely password that has not been queried by any of the other agents. 
Since this strategy reduces to querying as a group from the optimal list, the average number of queries completed by each agent is thus reduced by a factor of the number of agents, with respect to the case where a single agent queries alone. 
This requires the agents to be able to synchronize their queries, that is, there must be a knowledge of an ordering in which the agents make guesses. 

However, in many practical scenarios the adversarial agents are completely distributed and have limited communication with each other. 
One prime example is botnets, in which agents are often oblivious to the actions taken by other agents, and may have limited access to shared memory or synchronization tools. 
Owing to constraints of the physical computers in which these bots run, the speed, latency, and reliability of these agents is heterogeneous --- thus, perfect synchronization is unlikely. Note that even if a central agent distributes lists of possible guesses to the bots, such that the lists form a partition of all guesses, making sure no guess is repeated, the lack of synchronization may still render the process sub-optimal. We illustrate an example of synchronized and asynchronous attack in Fig~\ref{fig:intro}.
At one extreme, a complete lack of synchronization can be modeled by a worst-case optimization, in which the guesses of each agent come in the worst possible order. 

The goal of this paper is to study how much the lack of synchronization, as described above, might affect the overall number of queries that are made until the game ends. 
We discuss why deterministic strategies cannot perform well in this paradigm, while on the other hand, a simple randomized strategy in which all the guesses are drawn i.i.d. from a certain distribution asymptotically achieves the same optimal performance of a synchronous attack when guessing secrets that are long sequences drawn according to some types of distribution. 
This optimal guessing distribution is non-trivial, and, perhaps surprisingly, it is not the original password generating distribution $P_X$. It is a tilted distribution from $P_X$, where the tilt exponent depends on the moment of guesswork of interest.
In other words, distributed and asynchronous agents can adopt a strategy for which the asymptotic number of total queries sent before a system breach is optimal, regardless of the ordering in which these queries are received, but this distribution is only optimal for a given moment of guesswork, and not optimal universally across all moments. 

For the sake of simplicity of discussion, we have made the following assumptions on the password generation process, as well as on the brute-force attack itself.
\begin{enumerate}
    \item Passwords are assumed to be strings of given length $n$. Note that in some applications, the brute-force attack takes place on private key of some fixed size, in which case the length of the secret key is often known.
    \item Passwords are assumed to be strings whose characters are generated i.i.d. from a distribution $P_X$.\footnote{We briefly mention generalizations to passwords generated according to an irreducible stationary Markov Chain in Remark~\ref{remark:markov} in Section~\ref{sec:synchronization}.}
    \item The common goal of the agents is to guess one given password, or string, \changed{a so called targeted attack}. In practice, there might be multiple accounts which undergo attacks simultaneously.
    \item The agents have no additional information about the users and may  construct guesses based solely on $P_X$\footnote{\changed{We also briefly discuss the presence of side-information during the attack, which an adversary might use to modify the distribution of the potential passwords at the end of Section~\ref{sec:synchronization}.}}.
\end{enumerate}
We believe that some of these assumptions could be relaxed and generalized using techniques from the literature, as discussed below. \changed{In addition, the i.i.d. setting, and the resulting asymptotic results, can be used as guidelines in designing systems even if the real system violates the memoryless assumption. For example, such results can be used to choose the minimal length of a password to secure a system.} Despite these assumptions, the insights gained from the model we study shed light on the robustness of brute-force attacks to asynchronization. To illustrate this claim, we have shown our results on an extract of the Adobe Leaked password dataset (see \cite{AdobeLeak} for a description of the dataset). In particular, we extracted the $10^4$ most likely passwords from a subset of 10 millions passwords in the data, and restricted our study to those passwords. We investigate the guesswork when the correct password is drawn according to the distribution $P_X$ as computed on this restricted sample of the data. We show in Figure~\ref{fig:adobe} the performance of a randomized strategy when using the optimal guessing distribution versus the naive distribution $P_X$, both in terms of expected number of guesses and in terms of probability of making less than a fixed number of guesses. Note that the true distribution $P_X$ performs well if one wishes to make only a small number of guesses, but eventually takes longer to reach a high probability. This is due to less frequent passwords, which are barely ever queried if guesses are drawn according to $P_X$. The guessing distribution which optimizes the average number of guesses increases the probability of querying the less likely passwords, as those passwords represent the main computational burden on the adversary when they occur.
 
\begin{figure}
    \centering
    \begin{subfigure}{.45\textwidth}
    \includegraphics[width=.9\textwidth]{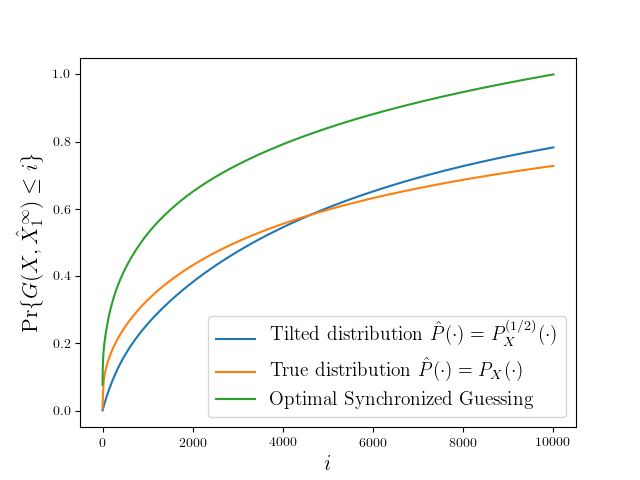}
    \caption{Probability of finding the password in fewer than $i$ queries. In a synchronized attack, the passwords has to be found after at most $|\mathcal{X}| = 1e4$ queries. The blue and orange line correspond to i.i.d. guesses according to the distribution $\hat{P}$. }
    \end{subfigure} \quad
    \begin{subfigure}{.45\textwidth}
    \includegraphics[width=.9\textwidth]{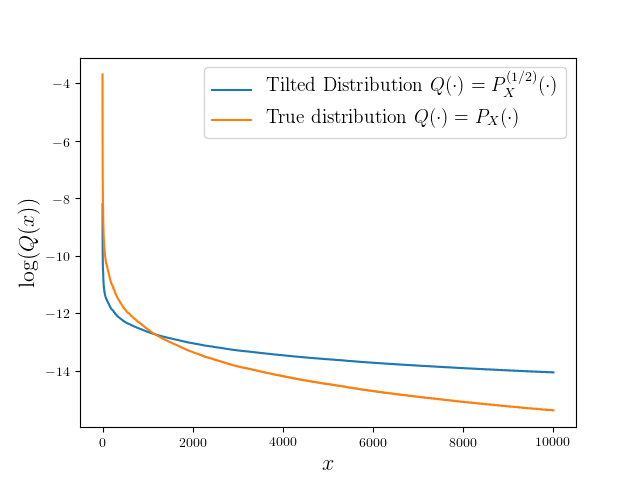}
    \caption{Log-probability mass function. Notice how the tilted distribution gives more weight to less likely symbols, as they correspond to the symbol which are the most costly for password guessing.}
    \end{subfigure}
    \caption{Experiments on a subset of Adobe Leak password data (only $10^4$ unique passwords kept). Despite the heavy tail of the distribution, a randomized strategy with some tilt improves the log expected number of guesses from 9.2 when using the naive distribution, to 8.8 when using the optimal tilt}
    \label{fig:adobe}
\end{figure}

\textbf{Related Work: }The problem of a cipher with a guessing wiretapper was considered in \cite{MerhavArikan}. The problem of guessing subject to distortion and constrained Shannon entropy were investigated in \cite{ArikanMerhav} and \cite{Beirami}, respectively. The above results have been generalized to ergodic Markov chains \cite{MaloneSullivan} and a wide range of stationary sources \cite{PfisterSullivan}. The problem of guessing under source uncertainty was investigated in \cite{Sundaresan}. The analysis of the guessing exponents, using large deviations theory, was considered in \cite{HanawalSundaresan}. In \cite{ChristiansenDuffy} it was shown that the guesswork satisfies a large deviation property and the rate function was characterized. Guessing a sequence given an erased version of the sequence was studied in \cite{christiansen2013guessing}, where the interplay between the large-deviations of the erasure process, and of the sequence generation were characterized. A brute-force attack where adversaries are interested in multiple passwords is discussed in \cite{christiansen2015multi}. A distributed attack model based on password hints was proposed in \cite{bracher2017guessing} and evaluated under guesswork metrics, and a wiretap system under guessing guarantees was studied in \cite{merhav1999shannon}. A geometric characterization of the guesswork was established in \cite{Beirami2} and expanded in \cite{beirami2018characterization}. Guesswork under an entropy budget was studied in \cite{rezaee2017guesswork}. Connections between guesswork and one-to-one coding were explored in \cite{kosut2017asymptotics}. \changed{Finally, applications of guesswork \cite{boztas2014renyi} to cryptographic guessing was studied in \cite{boztas2014renyi}, where oblivious or memoryless guessers were studied. The results of \cite{boztas2014renyi} are non-asymptotic, but very much related to our setting, as optimal i.i.d. guessing strategies both in terms of number of guesses and in terms of probability of success are studied, and a distributed attack scenario is also envisaged.}

\changed{The statistics of password generation were studied in \cite{wang2016implications,wang2017zipf,blocki2018economics,bonneau2012science}. {Password frequencies have been shown to follow closely variants of the Zipf's law distribution. In particular, the so-called \emph{CDF-Zipf's law} model introduced in \cite{wang2016implications,wang2017zipf} is a modification of the Zipf's law which captures the frequencies of passwords, both for very frequent passwords, and the tails, as exhibited by the close empirical fit to multiple password datasets (see \cite{wang2017zipf,wang2016implications,blocki2018economics}). Note that an adversary can benefit greatly from the the non-uniformity of these distributions to design more powerful brute-force attacks. Indeed, Guesswork, and other related notions of security related to brute-force attacks are also studied in \cite{wang2016implications,blocki2018economics,bonneau2012science}.} A special case of brute-force attack is given by \emph{targeted attacks}, in which the adversary uses the personal information of an user in his guessing strategy, see e.g. \cite{wang2016targeted}. Works such as \cite{wang2018security,das2014tangled} empirically demonstrate the threat of these targeted attacks, as most users chose their passwords according to some personal information which an adversary might have easy access to (e.g. birthdays, names of family members, locations, or simply password reuse) .}
 %\textcolor{red}{Complete}

\textbf{Main Contributions:} We define a min-max formulation that models a worst case asynchronous attack from the attacker's perspective, and show that a randomized strategy in which each guess is drawn i.i.d. from a certain distribution achieves the same asymptotic performance (in the length of the password sequence $n$) as an optimal synchronized attack. This optimal distribution is non-trivial; performing guesses according to the distribution from which the password was generated yields a strategy that is exponentially worse than the optimal guessing distribution. In fact, the optimal choice is a tilted distribution, where the tilt parameter is chosen depending on the moment of guesswork to be optimized. We also discuss optimal strategies when the benchmark is to maximize the probability of success of an attack with a fixed number of overall queries, and show that an i.i.d. guessing strategy again has optimal performance asymptotically. The optimal distribution is again a tilted distribution, where the tilt depends on the number of queries allowed. Together these results indicate that there is no loss in performance (asymptotically in $n$) when performing an asynchronous attack. 

The paper is organized as follows. In Section~\ref{sec:notation}, we establish some notation and provide a brief background on the guessing problem. We discuss the impact of synchronization under the number of guesses in Section~\ref{sec:synchronization} and then under the probability of a system breach with a fixed number of queries in Section~\ref{sec:prob_succ}.

\textbf{Previous Publication:} In a conference publication \cite{huleihel2017guessing},
we studied the problem of a memoryless guesser, and derived
some of the technical results appearing in the present paper
for binary sources, and integer moments $ \rho > 0$. Although
not directly related to botnets and asynchronous attacks, the
technical results introduced in that paper are at the core of the
analysis of an asynchronous distributed attack. In particular,
Lemma~2, Theorem~2 and Theorem~3 are present, although
restricted to binary sources and integer moments, in \cite{huleihel2017guessing}.
However, the connection to the asynchronous botnet problem
is novel and made explicit here. In particular, the min-max
formulation of Section III, along with Theorem~1, numerical
results, and extensions of our previous results to general finite
alphabets and non-integer moments distinguish this work from
our previous publication. Finally, Corollary~2 which characterizes
the loss in universality of distributed asynchronous
attacks, is also novel.

\section{Notation and Background}\label{sec:notation}

Throughout this paper, scalar random variables (RVs) will be denoted by capital letters, their sample values will be denoted by the respective lower case letters, and their alphabets will be denoted by the respective calligraphic letters, e.g. $X$, $x$, and $\mathcal{X}$, respectively. 
We also use the notation $\mathbf{X}_n$ to designate the sequence of RVs $(X_1,\ldots,X_n)$, and may drop the subscript when the size of the sequence considered is clear from the context, e.g., $\mathbf{X}$.
The set of all $n$-vectors with components taking values in a certain finite alphabet, will be denoted by the same alphabet superscripted by $n$, e.g., $\mathcal{X}^n$. Probability distributions will be denoted by the letters $P$ and $Q$, with subscripts that denote the names of the random variables involved along with their conditioning, if applicable, following the customary notation rules in probability theory. 
For example, $Q_{XY}$ stands for a generic joint distribution $\{Q_{XY}(x,y),\;x\in\mathcal{X},\;y\in\mathcal{Y}\}$, $P_{Y|X}$ denotes the matrix of single-letter transition probabilities, and so on.

The expectation operator will be denoted by $\mathbb{E}\ppp{\cdot}$, and when we wish to make the dependence on the underlying distribution $Q$ clear, we denote it by $\mathbb{E}_Q\ppp{\cdot}$. The Kullback-Liebler (KL) divergence between two probability measures $P$ and $Q$ will be denoted by $D(P||Q)$.
For entropies, it will be convenient to write explicitly the distributions, e.g. $H(P_X)$.
When dealing with binary random variables we may use the short-hand notation $H(p)$, where it is understood that it refers to the Shannon entropy over a Bernouilli distribution parametrized by $p$. A similar notation will be used for divergences, e.g., $D(p_1\|p_2)$.

For a given vector $\mathbf{x}_n$, let $\hat{P}_{\mathbf{x}_n}$ denote the empirical distribution, that is, the vector $\{\hat{P}_{\mathbf{x}_n}(x),~x\in{\mathcal{X}}\}$, where $\hat{P}_{\mathbf{x}_n}(x)$ is the relative frequency of the letter $x$ in $\mathbf{x}_n$. 
Let $T(P_X)$ denote the type class associated with $P_X$, that is, the set of all sequences $\mathbf{x}_n$ for which $\hat{P}_{\mathbf{x}_n}=P_X$. 

The cardinality of a finite set $\mathcal{A}$ will be denoted by $\abs{\mathcal{A}}$, its complement will be denoted by $\mathcal{A}^c$. 
The probability of an event $\mathcal{E}$ will be denoted by $\Pr\left\{\mathcal{E}\right\}$. 
For two sequences of positive numbers, $\left\{a_n\right\}$ and $\left\{b_n\right\}$, the notation $a_n\doteq b_n$ means that $\left\{a_n\right\}$ and $\left\{b_n\right\}$ are of the same exponential order, i.e., $n^{-1}\log a_n/b_n\to0$ as $n\to\infty$, where logarithms are defined with respect to (w.r.t.) the natural basis, that is, $\log\left(\cdot\right) = \ln(\cdot)$. 
Finally, for a real number $x$, we denote $[x]_+ \triangleq \max\{0,x\}$.

\noindent \textbf{Guessing Functions and Strategies: } A (possibly randomized) guessing strategy is a sequence $\hat{X}_{1}^\infty \triangleq \{\hat{X}_k(P_X) : k \geq 1\}$, where $\hat{X}_k(P_X) \in \mathcal{X}$, is independent of the realization $X$ but may depend on $P_X$. 
In other words, $\hat{X}_{1}^\infty$ is the list of guesses the attacker will use one after the other when trying to guess $X$.
The corresponding guessing function, $G(X,\hat{X}_1^\infty)$, defined as
\begin{align}
G(X,\hat{X}_1^\infty) \triangleq \inf \left\{k \geq 1 : \hat{X}_k(P_X) = X \right\},
\end{align} 
represents the number of queries before reaching $X$. 
The $\rho$-th moment of the number of guesses is thus given by $\mathbb{E}[ G(X,\hat{X}_1^\infty)^\rho ]$, where the expectation is taken over the distribution $P_X$ and the randomness inherent in the guessing strategy $\hat{X}_{1}^\infty$. 
The $\rho$-th moment \emph{guesswork} of a source $X\sim P_X$ is given by
\begin{align}
	\min_{\hat{X}_1^\infty}\mathbb{E}\left[G(X, \hat{X}_1^\infty)^\rho\right],
\end{align}
where the minimization is over all guessing strategies. 
In particular, the first moment, i.e. $\rho = 1$, corresponds to the average number of guesses that an adversary would have to perform before guessing the correct $X$.\footnote{Although the most relevant moment to consider for practical purposes is the expectation, i.e. $\rho = 1$, we consider here a more general quantity as to be consistent with existing literature on guesswork.}
It was shown in \cite{Messay} that, without any constraint on the set of possible guessing strategies, the optimal guessing strategy is obtained by ordering the symbols in $\mathcal{X}$ by decreasing order of $P_X$-probabilities, with ties broken arbitrarily, resulting in a deterministic strategy $\{\hat{x}_k(P_X) : k \geq 1\}$.
The resulting guessing function, denoted by $G^*(X)$, represents the position of $X$ in the optimal list, i.e. the list of symbols ordered from most likely to least likely. \footnote{Note that the optimal strategy $G^*(X)$ is deterministic, therefore the dependence on $\hat{X}_1^\infty$ is dropped in the notation. } 
The problem of bounding the expected number of guesses was investigated in \cite{Arikan}. Specifically, among other things, it was shown that for any $\rho\geq0$, and any guessing function $G(\cdot)$,
\begin{align}
\bE\pp{G(X)^\rho}\geq(1+\log\abs{\calX})^{-\rho}\pp{\sum_{x\in\calX}P_X(x)^{\frac{1}{1+\rho}}}^{1+\rho}.
\end{align}
On the contrary, the optimal guessing function, satisfies\footnote{An improved bound by a factor of 2 was reported in \cite{Boztas0}.}
\begin{align}
\bE\pp{G^*(X)^\rho}\leq\pp{\sum_{x\in\calX}P_X(x)^{\frac{1}{1+\rho}}}^{1+\rho}.
\end{align}
Finally, letting $\mathbf{X} = (X_1,X_2,\ldots,X_n)$ be a sequence of independent and identically distributed (i.i.d.) random variables over a finite set, and letting $G^*(\mathbf{X})$ denote the optimal guessing function of a realization of $\mathbf{X}$, it was shown that \cite[Proposition 5]{Arikan}
\begin{align} 
E_\rho(P_X) \triangleq \lim_{n\to\infty}\frac{1}{n}\log\bE\pp{G^*(\mathbf{X})^\rho} = \rho\cdot H_{\frac{1}{1+\rho}}(X_1),\label{ArikanResult}
\end{align}
where $H_\alpha(X)$ is the R\'{e}nyi entropy of order $\alpha$ ($\alpha>0$, $\alpha\neq1$), defined as
\begin{align}
H_{\alpha}(X)\triangleq \frac{1}{1-\alpha}\log\pp{\sum_{x\in\calX}P_X(x)^\alpha}.
\end{align}
Note that the function $E_\rho(P_X)$ simply quantifies the exponential growth of the guesswork, as $n \to \infty$. We note that \eqref{ArikanResult} gives an \emph{asymptotic} operational characterization/meaning to R\'{e}nyi entropy of order $0\leq\alpha\leq1$.

\section{Asynchronous Brute-Force Attack} \label{sec:synchronization}

In this section, we discuss synchronization when multiple agents aim to breach a secured system. 
Recall that we say that distributed agents are synchronized if they know in which order every agent's queries will be received by Alice. 
In this case, they can query from the optimal list as a group, \emph{i.e.}, the first query received is the most likely symbol, etc. In other words, full synchronization means they can all share a single (optimal) list, and a \emph{pointer} to this list advancing after each new guess.
As a result, the total number of queries sent is the same as the optimal single agent guesswork, namely, \eqref{ArikanResult} is achieved, while the individual computational burden on each agent is reduced since the queries are divided among agents. 
Further, even if the number of adversaries grows exponentially\footnote{Note that in practice, the number of agents usually needs to grow since most secured systems include a mechanism which blocks IP addresses after a given number of password attempts. Thus, if a single agent can only make $k$ queries, there must be at least $\lceil |\mathcal{X}|^n/k \rceil$ agents to guarantee that a password of length $n$ will be found. } with the length of the password $n$, the total number of queries remains the same\footnote{\changed{Note that in this work we use the total number of queries as the main metric for computational effort, as opposed to e.g. \cite{boztas2014renyi} where the average number of guesses \emph{per agent} is characterized.}}. 

Instead, if agents do not know in which order the queries are delivered, they must adopt a strategy which performs well under any such ordering. 
In particular, we shall adopt a worst-case approach in which the goal is to minimize the number of queries in the worst ordering. 
Specifically, let $\mathbf{X}$, an i.i.d. sequence of length $n$ generated from $P_X$, be the sequence to be guessed, and let $\{\hat{\mathbf{X}}^{(a)}_k : k \geq 1\}$ be the strategy of agent $a \in \mathcal{A}$, where $\mathcal{A}$ is a, possibly
infinite, countable set. 
Again, we shall be interested in the regime where $|\mathcal{A}|$ grows at least exponentially fast with $n$, and the goal is to characterize number of queries made in total. 
We let the permutation $\pi: \mathbb{N}^+ \to \mathcal{A} \times \mathbb{N}^+$ denote the ordering in which the queries are received, i.e., $\pi(i) = (a_i,k_i)$ means that the $i$-th query received is $\hat{X}_{k_i}^{(a_i)}$.
Denote by $\Pi$ the set of all such possible orderings. 
Under an ordering $\pi$, Alice receives the sequence of queries $\pi(\hat{\mathbf{X}}_1^\infty) \triangleq \{\hat{\mathbf{X}}_{k_i}^{(a_i)}: i \geq 1  \}$. 
Note that this permutation allows reordering of guesses of a given agent $a \in \mathcal{A}$ which may be received in any arbitrary order.
For some fixed strategies $\{\hat{\mathbf{X}}_k^{(a)} : k \geq 1\}$, the worst ordering in terms of guesswork is thus given by
\begin{align}
\sup_{\pi \in \Pi} \mathbb{E} \ppp{G(\mathbf{X}, \pi(\hat{\mathbf{X}}_1^\infty))^\rho }. \label{eq:worst_case}
\end{align}
The goal of the agents is to minimize the worst-case number of queries, or, in other words, solve the min-max problem
\begin{align}
	\inf_{\{\hat{X}^{(a)}_k , k \geq 1\} \text{ for } a \in \mathcal{A}} \  \sup_{\pi \in \Pi} \  \mathbb{E} \ppp{G(\mathbf{X}, \pi(\hat{\mathbf{X}}_1^\infty))^\rho }. \label{eq:min_max}
\end{align}
%A solution to this problem would mean that each agent can prepare its own list, to be used at his turn to guess, without synchronizing it with the other agents, and the worst case guesswork would be minimized.
The main result of this section, presented below, characterizes the asymptotic exponent of \eqref{eq:min_max}, as $n \to \infty$. The proof of this result, along with the associated lemmas, are given after some discussion.
\begin{theorem}\label{thm:min_max}
	For $\mathbf{X}_n$ an i.i.d. sequence according to $P_X$, and $\{\hat{\mathbf{X}}_k^{(t)}, k \geq 1 \}$ sequences of guesses which are independent over $a \in \mathcal{A}$, we have the following
	\begin{align}
	&\lim_{n\to\infty} \frac{1}{n} \log \left(\inf_{\{\hat{\mathbf{X}}^{(t)}_k : k \geq 1\}} \; \sup_{\pi \in \Pi} \;\mathbb{E}\ppp{G(\mathbf{X}_n, \pi(\hat{\mathbf{X}}_1^\infty))^\rho}\right) \nonumber\\
	&\ \ \ \ \ \ \ = \lim_{n\to\infty} \frac{1}{n} \log \mathbb{E}\ppp{G^*(\mathbf{X}_n)^\rho} \nonumber \\ 
	&\ \ \ \ \ \ \ = \rho \cdot H_\frac{1}{1+\rho}(X) .
	\end{align}
\end{theorem}
Note that guesswork measures the \emph{total number of guesses made by the agents.} Thus it is clear that with full synchronization among the agents this value will not depend on $|\mathcal{A}|$. In a sense, dependence on $|\mathcal{A}|$ for a certain scheme would indicate a \emph{lack of synchronization}, as it would suggest that queries are repeated by the agents. Surprisingly, Theorem~\ref{thm:min_max} states that even under a worst-case assumption, there exist a strategy under which the guesswork does not depend on $|\mathcal{A}|$ and is similar to the fully synchronous case.
The above result and \eqref{ArikanResult} show that \emph{synchronization is not necessary to achieve the asymptotic optimal guessing performance}. 
This can be equivalently formulated by an achievability strategy, and a converse. The converse result is trivial, as the performance of the synchronized strategy $\mathbb{E}\ppp{G^*(\mathbf{X})}$ upper bounds \eqref{eq:min_max}.

\begin{lemma}[Converse]\label{lem:converse}
    For any strategy $\hat{\mathbf{X}}^\infty$, 
    \begin{align}
    \inf_{\{\hat{X}^{(t)}_k , k \geq 1\} \text{ for } a \in \mathcal{A}} \  \sup_{\pi \in \Pi} \  \mathbb{E} \ppp{G(\mathbf{X}, \pi(\hat{\mathbf{X}}_1^\infty))^\rho} \geq \mathbb{E}\ppp{G^*(\mathbf{X})}. \end{align}
\end{lemma}
%In other words, regardless of the ordering in which the queries are received, and the number of adversaries, the asymptotic number of guesses that are made before the correct password is found is the same as in a fully synchronized attack.
We now turn to finding an appropriate strategy which would match this converse bound.
Let us first examine a naive solution to this problem. 
Consider the strategy which consists in letting each agent construct the optimal list and query it individually, that is $X_1^{(a)}$ is the most likely symbol for all $a \in \mathcal{A}$, $X_2^{(a)}$ the second most likely symbol, etc. 
It is easy to see that \eqref{eq:worst_case} would evaluate to a quantity which grows with the number of agents $|T|$.
Indeed, many queries are duplicated, and thus the overall number of queries grows with $|\mathcal{A}|$, without even reducing the computational burden on each adversary since they all must query the same password strings.
Note that this remains true if one considers a less stringent worst-case analysis, by for example, letting the guesses of each of the agent to be consistent among themselves, i.e. the permutation does not change the relative order of the guesses of each agent.

If instead the agents agree on a partition of the guesses before the attack, in a way such that no two guesses are repeated, then the correct password is queried by one unique agent. Again, it is easy to see that the worst-case analysis yields a quantity which grows with $|\mathcal{A}|$, even though it cannot grow beyond $|\mathcal{X}|^n$, as every unique password is queried at most once. In particular, if $|\mathcal{A}| = |\mathcal{X}|^n$, then the worst-case analysis achieves its upper-bound. Note that these observations are not only an artifact of the worst-case analysis, but rather a consequence of the deterministic nature of the queries.

This motivates us to study randomized strategies. 
In particular, we consider guesses, which are randomly and independently drawn according to a specific distribution, independent from each other, and identically distributed. 
We then study this optimal distribution in terms of the expected moments of guesswork.
Consider first a scalar $X \in \mathcal{X}$, generated from $P_X$.
We let $\{\hat{X}^{(a)}_k, k \geq 1\}$ be an i.i.d. process with respect to $\hat{P}(\cdot)$, for all $a \in \mathcal{A}$. 
For a given $\rho > 0$, we define the quantity
\begin{align}
V_\rho(X,\hat{X}_1^\infty) & \triangleq {G(X,\hat{X}_1^\infty) + \rho - 1 \choose \rho}, %V_\rho(X,\hat{X}_1^\infty)&\triangleq \frac{1}{\rho!}\prod_{l=0}^{\rho-1}(G(X,X_1^\infty)+l),
\end{align}
where ${x \choose y}$ is the generalized binomial coefficient defined in terms of the Gamma function $\Gamma(\cdot)$, i.e.
\begin{align}
{x \choose y} = \frac{\Gamma(x + 1)}{\Gamma(y + 1)\Gamma(x - y + 1)}.
\end{align}
In particular, $V_1(X,\hat{X}_1^\infty) =G(X,X_1^\infty)$. The motivation for this definition of $V_\rho(X,\hat{X}_1^\infty)$ will be made clear in the proof of Lemma~\ref{th:1}, where it allows us to compute a particular infinite sum neatly. Note that for large $G(X,\hat{X}_1^\infty)$ and fixed integer $\rho$, Stirling's approximation of the binomial coefficient directly gives $V_\rho(X,\hat{X}_1^\infty) \approx G(X,\hat{X}_1^\infty)^\rho / \rho!$, therefore $V_\rho(X,\hat{X}^\infty_1)$ approximates the behavior of the guesswork moment $G(X,\hat{X}_1^\infty)^\rho$, up to some factor.

We are interested in the following optimization problem
\begin{align}
\bE\{V^*_\rho(X,\hat{X}_1^\infty)\}\triangleq\inf_{\hat{P}\in\calP}\bE\{V_\rho(X,\hat{X}_1^\infty)\},\label{optimization}
\end{align}
where $\calP$ is the probability simplex and $\{\hat{X}_k : k \geq 1 \}$ is generated i.i.d. from $\hat{P}$. We let $\hat{P}^*_\rho$ designate the minimizer. The following Lemma is the main ingredient in proving an achievability and thus Theorem~\ref{thm:min_max}.
\begin{lemma}\label{th:1} For any  $\rho\geq1$,
	\begin{align}
	\log\bE\{V^*_\rho(X,\hat{X}_1^\infty)\} &= \rho \cdot H_{\frac{1}{1+\rho}}(X),
	\end{align}
	and for any $x\in\calX$,
	\begin{align}
	\hat{P}^*_\rho(x) = \frac{P_X(x)^{\frac{1}{1+\rho}}}{\sum_{x'\in\calX}P_X(x')^{\frac{1}{1+\rho}}}.\label{tiltedDis}
	\end{align}
\end{lemma}
Before providing the proof of Lemma~\ref{th:1} we briefly discuss our result. First, we note that contrary to \eqref{ArikanResult}, the above result provides an \emph{exact} operational meaning for R\'enyi entropy $H_\alpha(X)$ of order $\alpha>0$. It should be mentioned here that a similar interpretation for $H_{1/2}(X)$ was reported in \cite{Boztas1,HanawalSundaresan2,boztas2014renyi}. Also, we see that the optimal guessing distribution \eqref{tiltedDis} is simply the tilted distribution of $P_X$ of order $1/(1+\rho)$. It should be emphasized that, since the function $f(x) = x^{1/1+\rho}$ is monotone, creating an optimal list according to $\hat{P}_X$ yields the exact same list as if done according to $P_X$, see e.g. \cite{beirami2018characterization}. However, the list of guesses chosen i.i.d. according to $\hat{P}_X$ will be different from the one if guesses are made i.i.d. according to $P_X$. Indeed, letting $\hat{P}(x)=P_X(x)$ gives
$$
\log\bE\{G(X,\hat{X}_1^\infty)\} = \log\abs{\calX},
$$
which could be much worse than $\log\bE\{V^*_1(X,\hat{X}_1^\infty)\} = H_{1/2}(X)$. Namely, when one is allowed only to guess passwordds according to a certain distribution, independently, and without a list, then using the original distribution is strictly sub-optimal, and the tilted distribution should be used. This result is related to similar results from the source-coding literature in which a tilted distribution also appears as the solution of an optimization where longer codewords are penalized exponentially (see e.g. \cite{campbell1965coding,blumer1988renyi}). Finally, note that the result is not asymptotic. In particular, the randomized strategy can be used over an alphabet $\mathcal{X}$ where each $x \in \mathcal{X}$ corresponds to a password. This result is thus relevant to dictionary attacks, where queries are drawn according to a dictionary of possible passwords, and suggests that distributed dictionary attacks should use a guessing distribution which is a tilted version of the true distribution.

\begin{proof}[Proof of Lemma~\ref{th:1}]
First, note that given $X$, $G(X,\hat{X}_1^\infty)$ is a geometric random variable, and for $k\geq1$,
\begin{align}
\Pr\{G(X,\hat{X}_1^\infty)=k\}= \sum_{x\in\calX}P_X(x)(1-\hat{P}(x))^{k-1}\hat{P}(x).\nonumber
\end{align}
Then, for any $\rho > 0$, we have
\begin{align}
&\bE\{V_\rho(X,\hat{X}_1^\infty)\} = \sum_{m=1}^\infty {m + \rho - 1 \choose m - 1} \Pr\{G(X,\hat{X}_1^\infty)=m\}\nonumber\\
&= \sum_{x\in\calX}P_X(x)\hat{P}(x)\sum_{m=1}^\infty {m + \rho - 1 \choose m - 1}(1-\hat{P}(x))^{m-1}\nonumber.
\end{align}
In the following, we calculate the second summation term in the r.h.s. of the last equality. This is equivalent to calculating
$$
\sum_{m=1}^\infty{m + \rho - 1 \choose \rho}y^{m-1}.
$$
Note that, using the identity $\Gamma(x+1) = x\Gamma(x)$ recursively, we get that
\begin{align}
    \frac{\Gamma(m + \rho)}{\Gamma(\rho + 1)} &= (m+\rho - 1)\cdot(m + \rho - 2)\cdots(\rho+1) \nonumber\\
    & = (-1)^{m-1} (-\rho-1)\cdot (-\rho-2) \cdots (-\rho - m + 1) \nonumber\\
    & = (-1)^{m-1} \frac{\Gamma(-\rho)}{\Gamma(-\rho - m +1)},
\end{align}
which yields ${m + \rho - 1 \choose \rho} = (-1)^{m-1} {- \rho - 1 \choose m - 1}$, and together with the change of variable $k = m-1$ we obtain
\begin{align}
\sum_{m=1}^\infty{m + \rho - 1 \choose m - 1}y^{m-1} & = \sum_{k = 0}^\infty {- \rho - 1 \choose k} (-y)^k \\
& = (1 - y)^{-{\rho} - 1},
\end{align}
where the last equality follows from the binomial formula.
Thus,
\begin{align}
\bE\{V_\rho(X,\hat{X}_1^\infty)\}&=\sum_{x\in\calX}P_X(x)\hat{P}(x)\frac{1}{\hat{P}(x)^{1+\rho}}\nonumber\\
& = \sum_{x\in\calX}\frac{P_X(x)}{\hat{P}(x)^{\rho}}.\label{beflag}
\end{align}
Next, we minimize the last expression with respect to $\hat{P}\in\calP$. 
To this end, since \eqref{beflag} is convex in $\hat{P}$, $\hat{P}^*$ is given by the solution of (for $x\in\calX$)
$$
-\rho\cdot\frac{P_X(x)}{\hat{P}^*(x)^{\rho+1}}+\lambda=0,
$$
where $\lambda$ is a Lagrange multiplier, and thus,
\begin{align}
\hat{P}^*(x) = \frac{P_X(x)^{\frac{1}{1+\rho}}}{\sum_{x'\in\calX}P_X(x')^{\frac{1}{1+\rho}}}.\nonumber
\end{align}
On substituting this optimal distribution in \eqref{beflag} we finally get
\begin{align}
\bE\{V_\rho^*(X,\hat{X}_1^\infty)\}&=\sum_{x\in\calX}\frac{P_X(x)}{\hat{P}^*(x)^{\rho}}=\p{\sum_{x\in\calX}P_X(x)^{\frac{1}{1+\rho}}}^{1+\rho},\nonumber
\end{align}
as claimed.
\end{proof}

The previous lemma applies to a scalar RV $X$, but can be easily extended to sequences $\mathbf{X}_n$, as shown in the following corollary.
\begin{corollary}\label{cor:th1}
	Let $\mathbf{X}$ be a sequence of length $n$ generated i.i.d. from $P_X$. Then, we have,
	\begin{align}
		\lim_{n\to\infty} \frac{1}{n}\log \mathbb{E}\{V_\rho^*(\mathbf{X}, \hat{\mathbf{X}}_1^\infty)\} = \rho \cdot H_{\frac{1}{1+\rho}}(X).
	\end{align}
\end{corollary}

\begin{proof}
	Treating $\mathbf{X}$ as a random vector, a direct application of Lemma~\ref{th:1} yields
	\begin{align}
		\log\bE\{V^*_\rho(\mathbf{X},\hat{\mathbf{X}}_1^\infty)\} &= \rho \cdot H_{\frac{1}{1+\rho}}(\mathbf{X})\nonumber\\
		& = \p{\sum_{\mathbf{x}\in\calX^n}P_{\mathbf{X}}(\mathbf{x})^{\frac{1}{1+\rho}}}^{1+\rho}.\nonumber
	\end{align}
	The desired result follows by the additivity of the R\'enyi entropy.
\end{proof}

Note that when $\bX$ is generated i.i.d., tilting the marginal distributions and drawing symbols i.i.d., or tilting the entire product distribution result in the same optimal distribution.

\begin{remark}\label{remark:markov}
We note that the result above can be generalized to passwords $\bX$ which are generated according to an irreducible stationary Markov Chain. More precisely, let $U = (U_{ab})$  and $\gamma_a$, for $a,b \in \mathcal{X}$, be the stochastic matrix and stationary distribution of the Markov chain, respectively, so that
\begin{align}
    \mathrm{Pr}\ppp{\mathbf{X} = (x_1 \ldots x_n)} = \gamma_{x_1} \prod_{i = 1}^{n-1} U_{x_i x_{i+1}} \label{eq:markov}
\end{align}
Then, it was shown in \cite{MaloneSullivan} that
\begin{align}
    \lim_{n \to \infty} \log \mathbb{E}\ppp{G^*(\bX)^\rho} = \frac{1}{1+ \rho} \log \lambda,
\end{align}
where $\lambda$ is the Perron-Frobenius eigenvalue of the matrix with entries $W = (U_{ab}^{1/1+\rho})$ for $a,b \in \mathcal{X}$. Further, let $\ppp{l_a}$ and $\ppp{r_a}$ be the left and right eigenvectors of $W$ associated with $\lambda$, that is
\begin{align}
\sum_{a \in \mathcal{X}} l_a = 1, \quad \sum_{a \in \mathcal{X}} l_a W_{ab} = \lambda l_b, \quad \sum_{b \in \mathcal{X}} r_b W_{ab} = \lambda r_a.    
\end{align}
Analogously to the result of Corollary~\ref{cor:th1}, it can be shown that generating guesses $\hat{\mathbf{X}}$ according to a Markov Chain with entries $W_{ab}r_b/(\lambda r_a)$ achieves the asymptotic performance in \eqref{eq:markov}. A proof of this fact is outside the scope of this paper, but follows from steps outlined in \cite{MaloneSullivan} along with the proof of Lemma~\ref{th:1}. 
\end{remark}

\begin{remark}
In the standard guessing problem \cite{Arikan} Alice tries to guess $X$ using her knowledge of $P_X$. It is assumed that there are no constraints on the memory of Alice, namely, for each new guess, Alice knows her previous guesses, and thus she can adapt her new guess accordingly (i.e., she will not guess again a previous incorrect guess). 
\salman{The setting we consider here is equivalent to one in which Alice} cannot keep track of her guesses, but still knows the distribution $P_X$. 
It should be clear that in this case all that Alice can do is to present a sequence of i.i.d. guesses $\hat{X}_1,\hat{X}_2,\ldots$, drawn from some distribution $\hat{P}(\cdot)$, which shall be optimized in some sense. 
Lemma~\ref{th:1} can be equivalently interpreted as the performance of a memoryless, (or oblivious) attacker \changed{\cite{Boztas1,HanawalSundaresan2,huleihel2017guessing,boztas2014renyi}}.
\end{remark}
We are now ready to prove Theorem~\ref{thm:min_max}.
\begin{proof}[Proof of Theorem~\ref{thm:min_max}]
We start by noting that letting $\{\hat{\mathbf{X}}_k^{(t)}: k \geq 1 \}$ be an i.i.d. process distributed according to $\hat{P}^*$ (as defined in Lemma~\ref{th:1}) gives an upper bound on \eqref{eq:min_max}. 
We prove that two bounds match asymptotically, by showing that the exponent of the upper-bound is equal to $\rho \cdot H_{1/\rho+1}(X)$. 
Indeed, let $\{\mathbf{X}_k^{(t)}: k \geq 1 \}$ be an i.i.d. process distributed according to $\hat{P}^*$  for all $t \in T$. 
Then, it is evident that $\pi(\hat{\mathbf{X}}_1^\infty)$ is also an i.i.d. process distributed according to $\hat{P}^*$, for any permutation $\pi \in \Pi$. 
An application of Corollary~\ref{cor:th1} concludes the proof.
\end{proof}

Note that the optimal distribution from Lemma~\ref{th:1} depends on the moment $\rho$. Indeed, the larger $\rho$, the more we are penalized for passwords which are less frequent (which increase the work significantly). Therefore, the optimal strategy gives extra weight to less frequent symbols as to make sure that they are more likely to be chosen than what their probability suggests. We do so by tilting $P_X$ by $1/{(1+\rho)}$. Nevertheless, the optimal distribution, and thus guessing strategy, will change as a function of the guesswork moment $\rho$ of interest. This contrasts with the synchronous case, in which the optimal strategy consisting of querying the sequences from most likely to least likely is optimal universally for all moments $\rho$. This loss of universality is exploited in the following corollary, which characterizes the loss in using a distribution optimized for a moment $\rho > 0$, when measured in terms of a moment $\gamma \neq \rho$, and is illustrated for a binary source in Figure~\ref{fig:mismatch}.

\begin{corollary}\label{cor:2}
    Fix $\gamma > 0$, and let $\{\hat{X}_k: k \geq 1\}$ be an i.i.d. process generated according to $\hat{P}^*_\gamma(x)$. Then:
    \begin{align}
    \log \mathbb{E}\{V_\rho(X,\hat{X}_1^\infty)\} = \frac{\rho}{1 + \gamma} H_{\frac{\gamma - \rho + 1}{1 + \gamma}}(X) + \frac{\gamma \cdot \rho}{1 + \gamma}H_{\frac{1}{1 + \gamma}}(X)
    \end{align}
\end{corollary}
\begin{proof}
    The proof follows by substituting  $\hat{P}(\cdot) = \hat{P}^*_\gamma(\cdot)$ into \eqref{beflag}.
\end{proof}

\begin{figure}
\centering
\includegraphics[scale =.6]{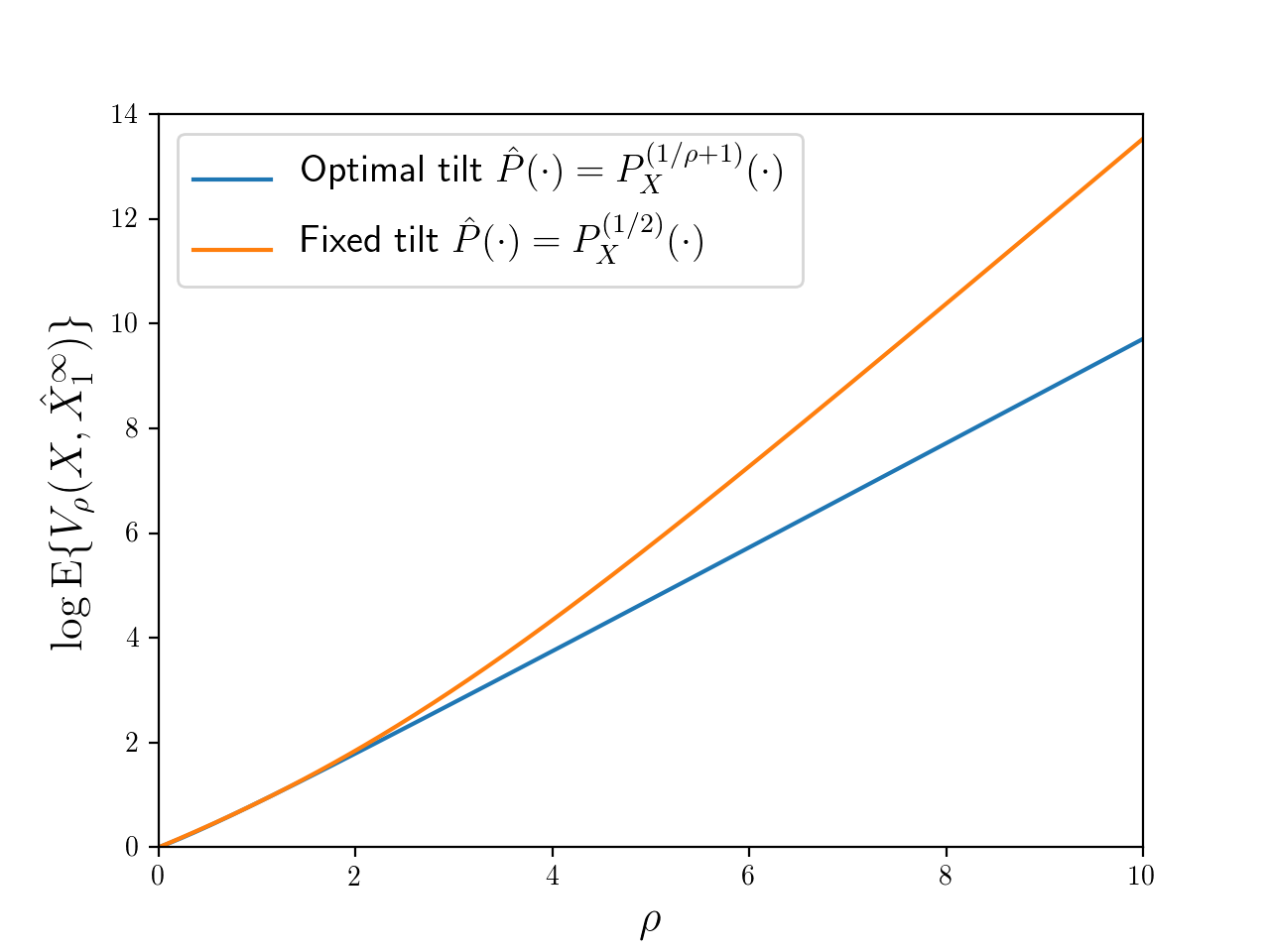}
\caption{This plots compares the performance of the randomized strategy as a function of the moment $\rho$. We compare the optimal strategy which depends on $\rho$, against a fixed tilted distribution ($\gamma = 1$ in Corollary~\ref{cor:2}), when $X\sim \mathrm{Ber}(1/5)$.}
\label{fig:mismatch}
\end{figure}

\begin{remark}[Zipf's distribution]
	\changed{
		We emphasize that Lemma~\ref{th:1} is a non-asymptotic result. As such, it can be readily used in the context of passwords generated according to a Zipf's law distribution of parameter $s$ for some $s \geq 0$ (also known as PDF-Zipf model \cite{wang2016implications}), i.e.,
		\begin{align}
			P_X(i) \triangleq \frac{1}{H_{m,s}} \cdot \frac{1}{i^s}
		\end{align}
		where $i = 1,\ldots, m$, and $H_{m,s}$ is the \emph{generalized} harmonic number defined as $H_{m,s} = \sum_{j = 1}^m \frac{1}{j^s}$. {As pointed out in the introduction, this family of distribution has been shown in the literature to be useful in modeling password distributions, where the parameter $s$ is dataset dependent. We refer to \cite{wang2016implications,wang2017zipf} for more details about the relevance of the Zipf's law in this setting.} Under this distribution, applying Lemma~\ref{th:1}, we obtain that the optimal i.i.d. guessing strategy is to generate guesses according to a Zipf's law of parameter $s/(\rho + 1)$. Further, we get that
		\begin{align}
			\log \mathbb{E}\left\{ V_\rho^*(X,\hat{X}_1^\infty) \right\} = (1 + \rho) \log H_{m,\frac{s}{1+\rho}} - \log H_{m,s}.
		\end{align}
		Note that this is worse than the optimal synchronized strategy which achieves $\log H_{m,(s - \rho)} - \log H_{m,s}$, for $s \geq \rho$, but can perform much better than picking the sub-optimal i.i.d. guessing distribution $P_{\hat{X}} = P_X$, which gives $\log m$. {Note that a similar result would hold for the so-called CDF-Zipf's law in \cite{wang2016implications}, i.e., when $P_X(i) = Ci^s - C(i - 1)^{s}$, for some  normalizing constant $C$ and parameter $0 \geq s \leq 1$. Namely, it is easy to show that the resulting optimal i.i.d. strategy is then according to the distribution $\hat{P}^*_\rho(i) = C' (i^s - (i-1)^s)^{\frac{1}{1 + \rho}}$, where $C'$ is once again a normalizing constant. }
	}
\end{remark}	

\begin{remark}[Targeted Attacks]
Lemma~\ref{th:1} can also be generalized to the case of availability of some side information $Y$ which is correlated with $X$. That is, $(X,Y)$ is now a pair of random variables with joint distribution $P_{XY}$. 
{This models targeted attacks \cite{wang2016targeted} where an adversary makes use of the additional information he possess about an user (e.g. personal
information, previously compromised passwords), as modeled by the side-information $Y$, to make guesses. Note that, as there are various kinds of side-information $Y$ (e.g., sister password, gender), each of which has a different role in impacting password creation, how to systematically employ such side-information $Y$ is subtle. We refer readers to  \cite{wang2016targeted} for a more precise treatment of targeted attacks, and the change in performance that results from them.} Then, assume that the guesser generates a sequence of guesses $\hat{X}_1,\hat{X}_2,\ldots$ which are i.i.d. \emph{given} $Y$, and distributed according to $\hat{P}_{X|Y}(\cdot\vert\cdot)$. As before, we define $G(X,\hat{X}_1^\infty|Y)\triangleq\inf\{k\geq1:\;\hat{X}_k(Y)=X\}$. Then, following the proof of Theorem~\ref{th:1} we can show that the optimal guessing distribution is 
\begin{align}
\hat{P}_{X|Y}^*(x|y) = \frac{P_{X|Y}(x|y)^{\frac{1}{1+\rho}}}{\sum_{x'\in\calX}P_{X|Y}(x'|y)^{\frac{1}{1+\rho}}}
\end{align}
for any $x\in\calX$ and $y\in\calY$, and
\begin{align}
\log\bE\{V^*_\rho(X,\hat{X}_1^\infty|Y)\} &= \rho \cdot H_{\frac{1}{1+\rho}}(X|Y),
\end{align}
where $H_\alpha(X|Y)$ is the conditional R\'enyi entropy of order $\alpha$, and $V^*_\rho(X,\hat{X}_1^\infty|Y)$ is defined as in \eqref{optimization} but with $G(X,\hat{X}_1^\infty)$ replaced by $G(X,\hat{X}_1^\infty|Y)$. \changed{This demonstrates that targeted attacks can also be performed in a distributed way by employing i.i.d. guesses from the distribution $P_{X|Y}(\cdot | Y)$. Note that this assumes that all distributed agents have access to the same side-information $Y$. A setting in which this does not hold true, i.e. agents may use different side-information $Y_i$, is outside the scope of this paper, but was studied in \cite{salamatian2017centralized}}. In particular, \cite{salamatian2017centralized} compares two mechanisms, one in which the agents do not share their side-information and attempt to breach the system independently, and one in which all the side-information is pooled.
\end{remark}

\section{Constraints on the Number of Guesses}\label{sec:prob_succ}
In Section~\ref{sec:synchronization}, we considered the case in which guesses are made until the correct sequence is found. 
In this section, we consider the case where adversaries can use only a fixed number of \salman{guesses} denoted by $J$. 
%For simplicity, we focus on the case of guessing $n$-length i.i.d. binary sequences $\calX^n=\ppp{0,1}^n$, and we assume that $J = \left\lceil 2^{n\alpha}\right\rceil$ \salman{guesses} are possible, for some $\alpha\geq0$. 
\changed{The goal of the adversary is then to maximize her probability of success within this fixed number of queries, both in the synchronized case \cite{Arikan}, as well as the asynchronous case. For synchronous guessers, the probability of success associated with the optimal strategy is given by}
\begin{align}
\mathrm{P}_{c,J}^{\text{synchr}} &= \sum_{x\in\calL}P_{X}(x),\nonumber
\end{align}
\changed{where $\mathcal{L}$ designates the set of the $J$ most likely elements according to $P_X$. For asynchronous guessers, one strategy consists in generating guesses $\hat{X}$ i.i.d. from a distribution $P_{\hat{X}}$, as was done in the previous section. This setting was precisely studied in \cite[Theorem~6]{boztas2014renyi}, where the optimal guessing distribution $P_{\hat{X}}$ was characterized as a function of the password distribution $P_X$ and of the number of guesses $J$.} Instead, in this work, we focus on the scenario of guessing $n$-length i.i.d. sequences, and we assume the adversaries make $J = \lceil \mathcal{X}^{n\alpha} \rceil$ total guesses.
We analyze the success probability in guessing the correct sequence and derive expressions which are exponentially tight as a function of $n$.
We consider both the synchronized case \cite{Arikan} as well as the asynchronous case. 

We start with synchronized guessers, and define the exponential rate of $\mathrm{P}_{c,J}^{\text{synchr}}$ as
\begin{align}
E_{c,\alpha}^{\text{synchr}}&\triangleq\liminf_{n\to\infty}-\frac{1}{n}\log\mathrm{P}_{c,J}^{\text{synchr}} \\
& = \liminf_{n\to\infty}-\frac{1}{n}\log \sum_{\mathbf{x} \in \mathcal{L}} P_{\mathbf{X}}(\mathbf{x}),
\end{align}
\changed{where again $\mathcal{L}$ represents the set of the $J$ most likely elements distributed according this time to the product distribution $P_{\mathbf{X}}$.}
The following result is an immediate application of the large deviation principle of Guesswork, shown in \cite{ChristiansenDuffy}.

\begin{theorem}[Theorem~3 in \cite{ChristiansenDuffy}]\label{th:2}
	For any $\alpha\in\pp{0,1}$,
	\begin{align}
	E_{c,\alpha}^{\text{synchr}} &= \min_{Q_X \in \mathcal{Q}(\alpha)}D(Q_X \| P_X), \label{eq:opti_prob_succ}
	\end{align}
	where $\mathcal{Q}(\alpha)$ is defined as:
	\begin{align}
	\mathcal{Q}(\alpha) &= \left\{Q_{X} : D(Q_{X}\| P_{X}) + H(Q_{X})< D(Q^*_{X}\| P_{X}) + H(Q^*_{X}) \right\},
	\end{align}
	with $Q^*_{X}$ being the solution of the optimization problem:
	\begin{equation}\label{eq:threshold_prob}
	\begin{aligned}
	& \underset{Q_{X}}{\text{minimize}}
	& & D(Q_{X} \| P_{X} ) + H(Q_{X}) \\
	& \text{subject to}
	& & H(Q_{X}) \geq \alpha
	\end{aligned}
	\end{equation}
	In particular, if $\alpha > H(P_X)$, then $E_{c,\alpha}^{\text{synchr}} = 0$.
\end{theorem}

Note that the average number of guesses, roughly $2^{nH_{1/2}(X)}$, is much larger than the required list size that drives $\mathrm{P}_{c,J}^{\text{synchr}}$ to one (exponentially). This great difference comes from the way atypical events are treated in each optimization. In the case of guesswork, an exponential price is payed for atypical events, since the number of queries will be exponential. For probability of error however, the scenario is closer to regular source coding in which the impact of atypical events is sub-exponential, meaning that the optimized quantity will necessarily be related to the typical events. 
Consider now the asynchronous case, and let $\{\hat{X}_k : k \geq 1\}$ be once again i.i.d. with distribution $P_{\hat{X}}$. In this case the probability of success is defined as
\begin{align}
\mathrm{P}_{c,J}^{\text{asynchr}} &\triangleq \Pr\ppp{G(\mathbf{X},\hat{\mathbf{X}}_1^\infty)\leq J}.
\end{align}
One can verify that
\begin{align}
\mathrm{P}_{c,J}^{\text{asynchr}} = \sum_{\mathbf{x}\in\calX^n}P_{\mathbf{X}}(\mathbf{x})\pp{1-(1-P_{\hat{\mathbf{X}}}(\mathbf{x}))^J}.\nonumber
\end{align}
Finally we define
\begin{align}
E_{c,\alpha}^{\text{asynchr}}\triangleq\liminf_{n\to\infty}-\frac{1}{n}\log\mathrm{P}_{c,J}^{\text{asynchr}}.\label{Exponent2}
\end{align}
While, in principle, the distribution $P_{\hat{\mathbf{X}}}$ can be optimized to maximize the probability of success, we will assume that this distribution is simply given by the tilted distribution of $P_{\mathbf{X}}$, namely, for some $\beta\geq0$, and any $\mathbf{x}\in\calX^n$,
\begin{align}
P_{\hat{X}}^{(\beta)}(x) \triangleq \frac{P_{X}(x)^\beta}{\sum_{x\in\calX}P_{X}(x)^\beta}.
\end{align}
We motivate this choice by the results of the previous sub-section, which showed that these tilted distributions were optimal in terms of the number of guesses.
We have the following result.

\begin{theorem}\label{th:3}
	For any $\alpha,\beta\geq0$,
	\begin{align}
	E_{c,\alpha}^{\text{asynchr}}(\beta)&=\min_{Q_X \in \mathcal{Q}(\alpha)}\left\{\vphantom{\pp{D(Q_X||P_{\hat{X}}^{(\beta)})+H(Q_X)-\alpha}_+}D(Q_X||P_X)+\pp{D(Q_X||P_{\hat{X}}^{(\beta)})+H(Q_X)-\alpha}_+\right\},
	\end{align}
	where $\pp{x}_+ \triangleq \max \{ x, 0\}$.
\end{theorem}

Using Theorem~\ref{th:3}, we obtain the following immediate result.
\begin{corollary}\label{cor:1}
\begin{align}
\min_{\beta\geq0}E_{c,\alpha}^{\text{asynchr}}&= \min_{Q_X \in \mathcal{Q}(\alpha)}D(Q_X \| P_X) \\
& = E_{c,\alpha}^{\text{synchr}}.
\end{align}
\end{corollary}
Corollary~\ref{cor:1} essentially proves that the tilted family is asymptotically optimal, and that there exist a unique optimal tilt $\beta$ for each size list $J = \lceil \mathcal{X}^{n \alpha}\rceil$. \changed{It follows from this that even though the optimization \eqref{eq:opti_prob_succ} is over a set of distributions $\mathcal{Q}(\alpha)$, the solution is always a tilted distribution $P^{(\beta)}_X$ for some $\beta \geq 0$ which depends on $\alpha$.}
\begin{proof}[Proof of Corollary \ref{cor:1}]
By definition, $\min_{\beta\geq0}E_{c,\alpha}^{\text{asynchr}}\geq0$. Then, for $\alpha \geq H(P_X)$, we see from Theorem~\ref{th:3} that by taking $Q_X=P_X$ and $\beta=1$, we have
\begin{align}
\min_{\beta\geq0}E_{c,\alpha}^{\text{asynchr}} &\leq \pp{H(P_X)-\alpha}_+= 0.
\end{align}
For $\alpha< H(P_X)$, we first note that by definition $\min_{\beta\geq0}E_{c,\alpha}^{\text{asynchr}}\geq E_{c,\alpha}^{\text{synchr}}$. Hence, due to Theorem~\ref{th:2} and Lemma~\ref{lem:types} in the appendix we may conclude that 
\begin{align} 
\min_{\beta\geq0}E_{c,\alpha}^{\text{asynchr}}\geq D(Q^*_X \| P_X),\label{con0}
\end{align}  
where $Q^*_X$ is the solution of the optimization
	\begin{equation}\label{eq:threshold_prob}
\begin{aligned}
& \underset{Q_X}{\text{minimize}}
& & D(Q_X\| P_X ) + H(Q_X) \\
& \text{subject to}
& & H(Q_X) \geq \alpha.
\end{aligned}
\end{equation}
On the other hand, by taking $Q_X = Q^*_X $, we have
\begin{align}
\min_{\beta\geq0}E_{c,\alpha}^{\text{asynchr}} & \leq D(Q^*_X\| P_X)+\min_{\beta\geq0}\pp{D(Q^*_X||P_{\hat{X}}^{(\beta)})}_+.\nonumber
\end{align} 
It is a simple exercise to verify that $Q^*_X$ is a tilted distribution, \emph{i.e.} there exist a $\tilde{\beta}$ such that $Q^*(x) = \frac{Q_X(x)^{\tilde{\beta} }}{\sum_{x'}Q_X(x')^{\tilde{\beta} }}$.
Letting $\beta = \tilde{\beta}$ gives
\begin{align}
\min_{\beta\geq0}E_{c,\alpha}^{\text{asynchr}} & \leq D(Q^*_X||P_X).\label{con1}
\end{align}
\salman{The result follows from combining \eqref{con0} and \eqref{con1}}.
\end{proof}

We next provide the proofs of Theorems \ref{th:2} and \ref{th:3}.

\begin{proof}[Proof of Theorem~\ref{th:3}]
For simplicity of presentation, we prove the theorem for binary sequences, i.e. $\mathcal{X} = \{0,1\}$, and assume that $1/2 \geq p \triangleq P_X(0)$. For any given sequence $x^n\in\calX^n$,
\begin{align}
\frac{1}{n}\log \hat{P}_{X^n}(x^n) &= -D(\hat{P}_{\mathbf{x}_n}||\bar{p}^\beta)-H(\hat{P}_{\mathbf{x}_n})\label{types2}
\end{align}
where $\hat{P}_{\mathbf{x}_n}$ is the empirical measure of a given sequence $x^n$, and $\bar{p}^\beta = \frac{p^\beta}{p^\beta + (1-p)^\beta}$. Then,
\begin{align}
\mathrm{P}_{c,J}^{\text{asynchr}}& = \sum_{x^n\in\calX^n}P_{X^n}(x^n)\pp{1-(1-\hat{P}_{\mathbf{x}_n})^J}\nonumber\\
& = \sum_{x^n\in\calX^n}2^{-n\p{D(\hat{P}_{\mathbf{x}_n}||p)+H(\hat{P}_{\mathbf{x}_n})}}\nonumber\\
&\ \ \ \ \ \ \ \ \ \ \ \times\pp{1-(1-2^{-n(D(\hat{P}_{\mathbf{x}_n}||\bar{p}^\beta)+H(\hat{P}_{\mathbf{x}_n}))})^J}.\nonumber
\end{align}
Letting $\mathcal{Q}_n$ denote the set of possible types, i.e. $\calQ_{n}\triangleq\ppp{0,1/n,2/n,\ldots,n/n}$ we obtain,
\begin{align}
\mathrm{P}_{c,J}^{\text{asynchr}}& = \sum_{q\in\calQ_{n,n}}\abs{T(q)}2^{-n\p{D(q||p)+H(q)}}\nonumber\\
&\ \ \ \ \ \ \ \ \ \ \times\pp{1-(1-2^{-n(D(q||\bar{p}^\beta)+H(q))})^J}\nonumber\\
& \doteq \sum_{q\in\calQ_{n,n}}\hspace{-0.2cm}2^{nH(q)}2^{-n\p{D(q||p)+H(q)}}2^{-n\pp{D(q||\bar{p}^\beta)+H(q)-\alpha}_+}\nonumber\\
& \doteq\max_{q\in\pp{0,1}}2^{-n\pp{D(q||p)+\pp{D(q||\bar{p}^\beta)+H(q)-\alpha}_+}}\nonumber
\end{align}
where the fourth equation follows from the fact that (see, e.g., \cite[Lemma 1]{AneliaMerhav}) if $a\in\pp{0,1}$, then $\frac{1}{2}\min\ppp{1,aM}\leq1-(1-a)^M\leq\min\ppp{1,aM}$. Thus, we have shown that
$$
E_{c,\alpha}^{\text{asynchr}}=\min_{q\in\pp{0,1}}\ppp{D(q||p)+\pp{D(q||\bar{p}^\beta)+H(q)-\alpha}_+}.
$$
\end{proof}

Together, Lemma~\ref{th:1} and Corollary~\ref{cor:1} imply that i.i.d. guesses can perform optimally, both in terms of the expected number of guesses, and in terms of the probability of success. Note that, analogous to Lemma~\ref{th:1}, the optimal distribution in Corollary~\ref{cor:1} depends on the parameter $\alpha$. As a result, asynchronous guessers can perform brute-force attacks as efficiently as synchronized guessers asymptotically, at the expense of universality. \changed{ Finally, it should be emphasized that the optimality of the tilted distribution is a by-product of the asymptotic treatment. Indeed, the results of \cite{boztas2014renyi} show that the optimal distribution in the non-asymptotic regime is not a tilted distribution of $P_X$, but rather a more involved functional of the password distribution. As such, our result does not follow from \cite{boztas2014renyi} in a straightforward way.}

\changed{
\begin{remark}[Probability of failure]
The above results characterized the probability of success of an adversary. In particular we demonstrated that a list size $J$ which is large enough (i.e., such that $\alpha> H(P_X)$) will have an exponent of success probability equal to 1, both under asynchronous and synchronous attacks. Note that this result can be strengthened by looking at the complementary probability of failure ${P}_{f,J}^{\text{synchr}}$ and $P_{f,J}^{\text{asynchr}}$. Again, in the i.i.d. setting, using essentially the same tools as for the probability of success, one can show that the exponents of the probability of failure for both synchronous and asynchronous attacks are the same, equal to $1$ when $\alpha < H(P_X)$, and decreasing as $\alpha$ grows. Similarly, the optimal guessing distribution for asynchronous guessers is a tilted distribution, where the tilt depends on the size of the list. 
\end{remark}
}

\changed{
	\begin{remark}[$J$-Guesswork]
		We briefly mention $J$-Guesswork, a related notion of computational security which was introduced in \cite{bonneau2012science} (denoted $\alpha$-Guesswork). While the usual Guesswork captures the average number of guesses necessary for a system breach, the average $J$-Guesswork, denoted by $\mathbb{E}[G_J(X)]$,  captures the average number of guesses for an adversary which performs no-more than $J$ queries, where $J$ is picked to guarantee a certain probability of success. As such, when $J = \mathcal{X}^{n}$, the $J$-Guesswork reduces to $\mathbb{E}[G(X)]$. We can rewrite the average $J$-Guesswork, as a sum of two terms, i.e.
		\begin{align}
		J \times \mathbb{P}(G(X) > J) + \sum_{i = 1}^J i \cdot P_X(i),
		\end{align}
		where the first term corresponds to the case where the attacker is unsuccessful and stops at $J$ guesses, and the second terms captures his average number of guesses otherwise. In the asymptotic regime where we look at passwords generated from the product distribution $P_{\mathbf{X}}$, and letting $J = \lceil |\mathcal{X}|^{n\alpha}\rceil$, for $\alpha > H(P_X)$, it follows from the remark above that the probability $\mathbb{P}(G(X) > J)$ goes to zero with an exponent $D(P_X^{(\beta)} \| P_X)$ for some unique $\beta \geq 0$, as long as $J$ is large enough (i.e. $\alpha>H(P_X)$). It is then easy to prove that, when $\alpha > H(P_X)$, the average $J$-Guesswork takes exponent
		\begin{align}
		\lim_{n\to\infty} \frac{1}{n} \log \mathbb{E}[G_\beta(\mathbf{X})] = \max \{ \alpha - D(P_X^{(\beta)} \| P_X), H_{1/2}(P_X) \}. \label{eq:tilt_alpha_guess}
		\end{align}
		When $J$ is too small, i.e. when $\alpha < H(P_X)$, then with high probability $G(X) > J$, and therefore the exponent is dominated by $J$ itself, that is $\lim_{n\to\infty} \frac{1}{n} \log \mathbb{E}[G_\beta(\mathbf{X})] = \alpha$. Note that these result hold true in an asynchronous setting as well. Indeed, picking guesses i.i.d. from a distribution $P_{\hat{\mathbf{X}}}$ such that it is equal to the tilted distribution which achieves the maximum in \eqref{eq:tilt_alpha_guess} gives the same exponent of $J$-Guesswork. Therefore, i.i.d. guesses perform asymptotically optimally with respect to $J$-Guesswork as well.
\end{remark}}

\section{Conclusion}

In this paper, we have studied the impact of synchronization on brute-force attacks. We showed that despite a lack of synchronization, and considering a worst-case ordering of the guesses, a randomized guessing strategy allows to achieve the optimal asymptotic performance, both in terms of average number of guesses, and in terms of probability of success after a given number of steps. As such, a solution which prevents repeated queries from a single IP is not enough, and in fact does not guarantee security against even completely asynchronous adversaries. This highlights the importance of password selection, as increasing the guesswork is the key to a secure password-based system.

The insights from these randomized strategies also applies to a single attacker who attempts to breach a system which is likely to be attacked by many other sources of attack. Against such as system, the attacker's strategy is analogous to one of a bot in a botnet. Indeed, since the system is likely to have been targeted by other attacks, the attacker might not want to follow his list in a deterministic way as to avoid repeating guesses from the other attackers. Using a randomized strategy does not hurt the performance asymptotically, but can prevent these repeated guesses.

A natural next step is to consider a distributed brute-force attack which aim at breaching any of $V$ password-secured accounts, rather than being aimed towards a single account. In this case, the computational effort will depend on the number of accounts which are under attack. More precisely, a brute-force attack directed against the accounts of $V$ members might be deemed successful as soon as $U$ of those accounts are compromised for $U \leq V$, regardless of which $U$ are compromised. The case where $U = 1$ corresponds to a classical brute-force attack directed at a multi-user system, while letting $U \geq 2$ models attacks on some distributed storage system, in which, because of the redundancy, some but not all of the servers should be compromised to access content. Additionally, once a system is compromised through sufficently many accounts, it may be much harder to reliably detect or counteract the actions of the attacker, .e.g., in the case of a Byzantine attack (c.f. \cite{Lamport} or \cite{perlman1988network}).
Generalizations of the standard Guesswork problem to this setting have been studied (see \cite{christiansen2015multi}), and establish the gain that arises from considering more accounts, especially when $U$ is much smaller than $V$. However, the optimal strategies in this case rely on a round-robin approach --- assuming the password generation process for all users is identical. More precisely, one should make password guesses to each account in turns, first making a guess for the first account, then the second, and so-on, until eventually successfully guessing the passwords of $U$ of the $V$ accounts. Generalizing such attacks to a distributed asynchronous case is of interest, and the subject of some future work.

\section*{Acknowledgment}

The authors are very grateful to Ken Duffy (Maynooth University) for many fruitful discussions and for providing the Adobe leaked password dataset used to generate Fig. 2.

\bibliographystyle{IEEEtran}
\bibliography{strings}

\appendices
\numberwithin{equation}{section}

\section{Additional Lemmas}

The following lemma relates the position of a sequence $\mathbf{x}_n$ in the optimal list, with the type of that sequence.
\begin{lemma}\label{lem:types}
	Let $\mathbf{x}_n$ be a i.i.d. generated sequence ,and consider the position of $\mathbf{x}_n$ in the optimal list according to $P_{X}$, \emph{i.e.} $G^*(\mathbf{x})$. 
	For a given $\alpha$, we have that $G^*(\mathbf{x}) < \lceil |\mathcal{X}|^\alpha \rceil $ if and only if the sequence $\mathbf{x}$ satisfy $\hat{P}_{\mathbf{x}} \in \mathcal{Q}(\alpha)$, where
	\begin{align}
	\mathcal{Q}(\alpha) &= \left\{Q_{X} : D(Q_{X}\| P_{X}) + H(Q_{X})< D(Q^*_{X}\| P_{X}) + H(Q^*_{X}) \right\},
	\end{align}
	with $Q^*_{X}$ being the solution of the optimization problem:
	\begin{equation}\label{eq:threshold_prob}
	\begin{aligned}
	& \underset{Q_{X}}{\text{minimize}}
	& & D(Q_{X} \| P_{X} ) + H(Q_{X}) \\
	& \text{subject to}
	& & H(Q_{X}) \geq \alpha
	\end{aligned}.
	\end{equation}
\end{lemma}

\begin{proof}
	Recall that $P_{X}(\mathbf{x}) = \exp \{-n \left(D(\hat{P}_{\mathbf{x}}\| P_{X}) + H(\hat{P}_{\mathbf{x}}) \right)\}$, and that the size of the type set $T(\hat{P}_{\mathbf{x}}) \doteq 2^{nH(\hat{P}_{\mathbf{x}})}$. 
	Let $\mathcal{Q}(\alpha)$ be the set of types of the sequences that are in the first $\mathcal{X}^{n\alpha}$ position in the list optimal list. Then, by definition of $\mathcal{Q}(\alpha)$:
	\begin{align}
		\sum_{Q_{X} \in \mathcal{Q}(\alpha)} 2^{nH(Q_X)} \doteq 2^{n\alpha}
	\end{align}
	An application of the method of types gives that the left-hand side evaluates to $2^{n \sup_{Q_{X} \in \mathcal{Q}(\alpha)}H(Q_X)}$, meaning that $\sup_{Q_{X} \in \mathcal{Q}(\alpha)}H(Q_{X}) = \alpha$.
	Thus, the threshold probability is given by the type that solves \eqref{eq:threshold_prob}, and any type that has lower probability must appears before in the list.
\end{proof}

\end{document}